\documentclass[runningheads,a4paper]{llncs}

\usepackage[utf8]{inputenc}
\usepackage{amssymb}
\usepackage{amsmath}
\usepackage{enumerate}
\usepackage{tikz}
\usepackage[ruled,vlined,linesnumbered]{algorithm2e}
\usepackage{lineno}

\newcommand{\hs}{\langle S \rangle}

\usepackage{url}
\urldef{\mailmitre}\path|{mitre@dcc.ufrj.br}|
\urldef{\mailrudini}\path|{rudini@lia.ufc.br}|

\newcommand{\keywords}[1]{\par\addvspace\baselineskip
\noindent\keywordname\enspace\ignorespaces#1}

\begin{document}

\mainmatter  

\title{Complexity aspects \\ of the triangle path convexity}

\author{Mitre C. Dourado\inst{1}\thanks{Partially supported by CNPq and FAPERJ, Brazil.}
\and Rudini M. Sampaio\inst{2}\thanks{Partially supported by CNPq (Universal 478744/2013-7), Brazil.}}

\institute{Universidade Federal do Rio de Janeiro, Rio de Janeiro, Brazil \\
\mailmitre
\and Universidade Federal do Ceará, Fortaleza, Brazil \\
\mailrudini
}

\maketitle

\begin{abstract}
A path $P = v_1, \ldots, v_t$ is a {\em triangle path} (respectively, {\em monophonic path}) of $G$ if no edges exist joining vertices $v_i$ and $v_j$ of $P$ such that $|j - i| > 2$; (respectively, $|j - i| > 1$). A set of vertices $S$ is {\em convex} in the triangle path convexity (respectively, monophonic convexity) of $G$ if the vertices of every triangle path (respectively, monophonic path) joining two vertices of $S$ are in $S$.
The cardinality of a maximum proper convex set of $G$ is the {\em convexity number of $G$} and the cardinality of a minimum set of vertices whose convex hull is $V(G)$ is the {\em hull number of $G$}.
Our main results are polynomial time algorithms for determining the convexity number and the hull number of a graph in the triangle path convexity.
\keywords{Convexity number, graph convexity, hull number, triangle path convexity}

\end{abstract}

\section{Introduction} \label{sec:intro}

Given a finite set $X$, a family ${\cal C}$ of subsets of $X$ is a {\em convexity on $X$} if $\varnothing, X \in {\cal C}$ and ${\cal C}$ is closed under intersections~\cite{Vel1993}. 
For every set $S$, we say that $S$ is a {\em convex set of ${\cal C}$} if $S \in {\cal C}$ and a {\em concave set of ${\cal C}$} if $X \setminus S \in {\cal C}$.
The {\em convex hull of $S \subseteq X$ in ${\cal C}$}, $\langle S\rangle _{\cal C}$, is the minimum convex set of ${\cal C}$ containing $S$.
Many convexities can be defined by a $2$-interval operator. A {\em $2$-interval operator} is a function $I: X \times X \rightarrow 2^X$ satisfying $a,b \in I(a,b)$ and $I(a,b) = I(b,a)$. Then, we say that a convexity ${\cal C}$ is {\em induced by an interval operator $I$} if for every convex set $C \in {\cal C}$ and elements $u,v \in C$, it holds $I(u,v) \subseteq C$.
Hence, we can also denote the convex hull of $S$ by $\langle S\rangle _I$.
This is the case of the most well studied graph convexities, in which the interval operator is generally defined using a family of paths ${\cal P}$. More specifically, the interval operator is defined as $I(u,v) = \{ w : w$ belongs to some path of ${\cal P}$ joining $u$ to $v\}$.
Sometimes, it will be useful to know the union of the intervals of all pairs of a set. Then, we define the {\em interval of a set $S$} as $[S]_I = \underset{u,v \in S}{\bigcup} I(u,v)$ if $|S| \geq 2$; and $[S]_I = S$ otherwise. If $[S]_I = X$ we say that $S$ is an {\em interval set of ${\cal C}$} and if $\langle S\rangle _I = X$ that $S$ is a {\em hull set of ${\cal C}$}.

In this work, we present some results on the triangle path~\cite{Changat1999} convexity.
Recall that a path $P = v_1, \ldots, v_t$ is a {\em triangle  path of a graph $G$} if no edges exist in $G$ joining pairs of vertices $v_i$ and $v_j$ of $P$ such that $|j - i| > 2$. Then a set of vertices $S$ of a graph $G$ is convex in the triangle path convexity, or {\em $t$-convex}, if the vertices of every triangle path of $G$, joining two vertices of $S$, are contained in $S$.

Other well-known graph convexities are the geodetic convexity~\cite{FarberJamison1986,pelayo2013}, the monophonic convexiy~\cite{Duchet-mono,EJ1985},
and the $P_3$ convexity~\cite{upper-Radon,Henning2013},
where the set of paths considered in the definition of the interval operator are the ``geodetic paths", ``minimal paths", and ``paths of order three", respectively.
For conciseness, sometimes we will use the corresponding symbols $t$ (triangle path convexity), $m$ (monophnic convexity), $g$ (geodetic convexity), and $P_3$ ($P_3$ convexity) for indicating the associated convexity instead of its entire name. For example, we will write ``$t$-hull set" when refering to a ``hull set in the triangle path convexity". Further, we will indicate the graph for which the convexity is associated in the cases that it is not obvious. For instance, if $S$ is a set of vertices of a graph $G$ that is subgraph of $G'$, then $[S]_g^G$ is the interval of $S$ in the geodetic convexity of $G$. 

The classical convexity invariants Carathéodory, Helly, and Radon numbers have already been determined for the triangle path~\cite{Changat1999} and the monophonic~\cite{Duchet-mono,JaNo1984} convexities. Our results are concerned on the following well known graph convexity problems:

\begin{enumerate}[{\em Problem} $1:$]
  \item To decide whether a set of vertices is convex;
  \item To compute the interval of a set of vertices;
  \item To compute the convex hull of a set of vertices;
  \item To find a maximum proper convex set of a graph ({\em convexity number of the graph});
  \item To find a minimum interval set of the graph ({\em interval number of the graph}); and
  \item To find a minimum hull set of the graph ({\em hull number of the graph}).
\end{enumerate}

All these six problems have already been considered in the geodetic~\cite{Araujo2013,Dourado2009,Ekim2012,gi}, monophonic~\cite{Costa2014,DPS-mhull}, and $P_3$~\cite{Centeno2011,Rudini} convexities.

The text is organized as follows. The results contained in Sections~\ref{sec:pre} to~\ref{sec:hullnumber} are relatively to the triangle path convexity. More specifically, in Section~\ref{sec:pre} we present some definitions, notations, and observe that some results for Problems $2$ and $5$ can be obtained using known results. 
In Section~\ref{sec:convexsets}, we present a characterization of convex sets that leads to a polynomial time algorithm for solving Problems $1$ and $3$ for general graphs.
In Section~\ref{sec:convexity}, we present a polynomial time algorithm for determining the convexity number of a graph, solving Problem $4$ for general graphs. It is worth to observe that, among the considered convexities, this is the only one in which this problem can be solved in polynomial time. 
In Section~\ref{sec:hullnumber}, we present a characterization of minimum hull sets, Problem $6$, which leads to a polynomial time algorithm for finding such a set.

\section{Preliminaries} \label{sec:pre}

We begin this section giving some useful definitions.
For a natural number $k$ denote $\{1, \ldots, k\}$ by $[k]$.
We consider finite, simple, and undirected graphs. For a graph $G$, its vertex and edge sets are denoted $V(G)$ and $E(G)$, respectively, whose cardinalities are the order and the size of $G$.
For $S \subseteq V(G)$, denote by $G - S$ the subgraph obtained by the deletion of the vertices of $S$ and by $G[S]$ the {\em subgraph of $G$ induced by $S$}.
If every two vertices of $S$ are adjacent, then $S$ is a {\em clique of $G$}; and if every two vertices are not adjacent, then $S$ is an {\em independent set of $G$}. A graph is {\em bipartite} if its vertex set can be partitioned into two independent sets.
We say that $S$ is a {\em separator of $G$} if there are non-empty sets $A,B \subset V(G) \setminus S$ such that every path in $G$, between some $a \in A$ and $b \in B$, contains a vertex in $S$.
And that $S$ is a {\em minimal separator for $v,w \in V(G)$} if $S$, but no proper subset of $S$, separates $v$ and $w$ in $G$. We say that $S$ is a {\em relative minimal separator for $G$} if there are vertices $v,w \in V(G)$ such that $S$ is a minimal separator for $v$ and $w$. 
Separators that are cliques are called {\em clique separators}. 
If $S$ is a clique separator of $G$, for every connected component $H$ of $G - S$, the subgraph $G[V(H) \cup S]$ is a {\em $S$-component of $G$}.
We say that $G$ is {\em reducible} if it contains a clique separator, otherwise it is a {\em prime}.
A {\em maximal prime subgraph of $G$}, or {\em $mp$-subgraph of $G$}, is a maximal induced subgraph of $G$ that is a prime.

Consider an ordering of the $mp$-subgraphs $F_1, \ldots, F_t$ of $G$ and then define $R_i = V(F_i) \cap ( V(F_1) \cup \ldots \cup V(F_{i-1}) )$, for $i \in [t]$. This ordering of the $mp$-subgraphs is a {\em $D$-ordering} if, for all $i \in \{2, \ldots, t\}$, there is a $p < t$ with $R_t \subseteq V(F_p)$. According to Theorem~2.5 of~\cite{Leimer1993}, there is a $D$-ordering for the $mp$-subgraphs of any graph. Further, Proposition 2.4 of~\cite{Leimer1993} says that every permutation of the $mp$-subgraphs of $G$ that is a $D$-ordering has the same family of $R_i$ sets. Then, we can define ${\cal R}(G) = \{ R_2 \cup \ldots \cup R_k \}$ and $R(G) = R_2 \cup \ldots \cup R_k$, taking as basis any $D$-ordering of the $mp$-subgraphs of $G$.
The following result contains important properties of the members of ${\cal R}(G)$.

\begin{theorem} {\em \cite{Leimer1993}} \label{thm:41}
A set $C$ of vertices of a graph $G$ is a clique and a relative minimal separator for $G$ if and only if $C \in {\cal R}(G)$.
\end{theorem}

Observe that the convex hull of a set can be obtained by applying the interval operator that induces the considered convexity. The procedure consists simply of testing if $S$ is a convex set and, for the negative answer, since there exists some pair of elements $u,v \in S$ such that its inverval is not contained in $S$, redefine $S$ as $S \cup I(u,v)$ and reapply until a convex set be obtained. This process will eventually converge because we are considering the ground set is finite.

Since every minimum path is an induced path, and the latter is also a triangle path, the following relation holds for every set of vertices $S$ of a graph
$[S]_g \subseteq [S]_m \subseteq [S]_t$ and, consequently, $\langle S \rangle_g \subseteq \langle S \rangle_m \subseteq \langle S \rangle_t$, where $g,m,$ and $t$ stand for the interval operators inducing the geodetic, monophonic, and triangle path convexities, respectively. 
Further, note that the monophonic and the triangle path convexites coincide on bipartite graphs.
This fact allows us to use known results to conclude that in the triangle path convexity Problems $2$ and~$5$ are NP-complete even for bipartite graphs.

\begin{theorem} {\em ~\cite{Costa2014}}
Let $u,v,w$ be vertices of a bipartite graph $G$. The problem of dedicing if $w$ belongs to the interval of $\{u,v\}$ in the monophonic convexity is NP-complete.
\end{theorem}

Then, the decision version of Problem 2 is NP-complete in the triangle path convexity even for bipartite graphs and sets of size two.

\begin{corollary} \label{cor:problem2}
Let $u,v,w$ be vertices of a bipartite graph $G$. The problem of dedicing if $w$ belongs to the interval of $\{u,v\}$ in the triangle path convexity is NP-complete.
\end{corollary}

Corollary~\ref{cor:problem2} has as consequence that the algorithm proposed above for computing the convex hull of a set using the interval operator, can not be used efficiently for the triangle path convexity, unless $P = NP$, as can be done for other convexities, like geodetic and $P_3$ convexities, in which the interval of a set can be computed in polynomial time even for general graphs. However, in the next section, we show how to compute the convex hull of a set in the triangle path convexity in polynomial time.

\begin{theorem} {\em ~\cite{Costa2014}}
Let $G$ be a bipartite graph. The problem of dedicing if there exists a set $S$ such that $[S]_m = V(G)$ for $|S| \leq 2$ is NP-complete.
\end{theorem}

Then, the decision version of Problem 5 is NP-complete in the triangle path convexity even for bipartite graphs and any fixed integer greater than or equal to two.

\begin{corollary} \label{cor:problem5}
Let $G$ be a bipartite graph. The problem of dedicing if there exists a set $S$ such that $[S]_t = V(G)$ for $|S| \leq 2$ is NP-complete.
\end{corollary}

\section{Convex sets and convex hulls} \label{sec:convexsets}

This section has two parts, one concerned to general graphs and other to prime graphs. The first one contains characterization of $t$-convex sets which leads to polynomial time algorithms for Problems $1$ and~$3$ defined in Section~\ref{sec:intro}. In the second one, these results are adapted for prime graphs and will be very useful in next sections.

\subsection{General graphs}

We begin presenting a characterization of $t$-convex sets.

\begin{theorem} \label{thm:tp-convexset}
A set of vertices $S$ of a graph $G$ is $t$-convex if and only if
there is no vertex outside of $S$ having two neighbours in $S$ and
there are no two non-adjacent vertices of $S$ having neighbours in a same connected component of $G - S$. 
\end{theorem}


We observe that the above result can be rewritten using characterizations of $m$-convex and $P_3$-convex sets. We recall the characterizations of such convex sets below.

\begin{theorem} {\em \cite{DPS-mhull}} \label{thm:m-convexset}
A set of vertices $S$ of a graph $G$ is $m$-convex if and only if there are no two non-adjacent vertices of $S$ having neighbours in a same connected component of $G - S$. Further, one can decide in $O(nm)$ if $S$ is $m$-convex.
\end{theorem}

It is clear from the definition that a set of vertices $S$ of a graph $G$ is $P_3$-convex if and only if there is no vertex outside of $S$ having two neighbours in $S$. Further, one can decide in $O(n^2)$ if $S$ is $P_3$-convex.
Using this observation and the last two results we have a characterization of $t$-convex sets which allows the recognition of such sets in polynomial time.

\begin{corollary} \label{cor:tconvexset}
A set of vertices $S$ of a graph $G$ is $t$-convex if and only if $S$ is $m$-convex and $P_3$-convex. Further, one can decide in $O(nm)$ if $S$ is $t$-convex.
\end{corollary}

Now, we show how to use this result for computing the $t$-convex hull of a set in polynomial time.

\begin{theorem} \label{thm:alg_mconvexhull}
Let $G$ be a graph and $S \subseteq V(G)$. The $t$-convex hull of $S$ can be computed in $O(n^2m)$ steps.
\end{theorem}

\begin{proof}
Using the algorithm that follows directly of Corollary~\ref{cor:tconvexset}, we test if $S$ is a $t$-convex set in time $O(nm)$.
In the affirmative case, we are done.
Otherwise, the test returns
either a vertex $v \not\in S$ such that $v$ has two neighbours in $S$
or a connected component $C$ and two vertices $u,v \in X$ such that both vertices have neighbours in $V(C)$.
In the former case, redefine $S$ as $S \cup \{v\}$.
It the latter case,
since any shortest path of an induced subgraph of $G$ is an induced path of $G$,
we look for a shortest path from $u$ to $v$ in the induced subgraph $G[V(C) \cup \{u,v\}]$.
Such induced subgraph and path $P$ can be found in time $O(n+m)$~\cite{Cormen2001}.
Then, redefine $S$ as $S \cup V(P)$.
Repeat this process until a $t$-convex set be obtained.
It is clear that the number of iterations is less than $n$. Hence, the overall complexity of this
algorithm is $O(n^2m)$.
\end{proof}

We remark that in~\cite{Changat1999} there exists a characterization of the $t$-convex hull of a set of vertices (Theorem 2.1 of~\cite{Changat1999}). However, no polynomial time algorithm to compute it seems to follow.

\subsection{Prime graphs}

The special case of prime graphs turns out interesting since its solution is easier for these problems and, as we will see in next sections, can be used for solving the general case.
Since every induced path is a triangle path, the following result also holds for the triangle path convexity.

\begin{theorem} {\em \cite{DPS-mhull}} \label{thm:hs-atom-mono}
Every pair of non-adjacent vertices is an $m$-hull set of a prime graph.
\end{theorem}

\begin{corollary} \label{cor:hs-atom-tp}
Every pair of non-adjacent vertices is a $t$-hull set of a prime graph.
\end{corollary}

Using Theorem~\ref{thm:tp-convexset} and Corollary~\ref{cor:hs-atom-tp} we can charactize $t$-convex sets of prime graphs in a simpler way.

\begin{theorem} \label{thm:atom-convex}
Let $G$ be a prime graph and $S \subset V(G)$. Then, $S$ is a $t$-convex set if and only if $S$ is a clique such that every vertex outside $S$ has at most one neighbour in $S$.
\end{theorem}

\begin{proof}
Let $G$ be a prime graph and $S \subset V(G)$ a $t$-convex set. 
First, suppose that $S$ is not a clique and let $u,v \in S$ such that $uv \not\in E(G)$. Then, by Corollary~\ref{cor:hs-atom-tp}, $\{u,v\}$ is a hull set of $G$, a contradiction. Hence $S$ is a clique.
Now, if some vertex outside $S$ had two or more neighbours in $S$, by Theorem~\ref{thm:tp-convexset}, $S$ would not be a $t$-convex set. Completing the proof of the sufficiency. The necessity is direct from Theorem~\ref{thm:tp-convexset}.
\end{proof}

This result has interesting consequences.

\begin{corollary} \label{cor:atom}
If $G$ is a prime graph, then any two $t$-convex sets of $G$ share at most one vertex.
\end{corollary}


These two results can be used to show an upper bound on the number of big non-trivial $t$-convex sets.

\begin{corollary} \label{cor:numberprime}
If $G$ is a prime graph or order $n$, then the number of non-trivial $t$-convex sets of $G$ with at least three vertices is less then $n$.
\end{corollary}


Corollary~\ref{cor:numberprime} implies that the number of $t$-convex sets of a prime graph $G$ of order $n$ is upper bounded by $2n+1$, since every vertex is a $t$-convex set and the sets $\varnothing$ are $V(G)$ are $t$-convex.
Next, we present a characterization of $t$-convex hulls of prime graphs.

\begin{lemma} \label{lem:car-prime}
Let $S$ be a set of vertices of a prime graph $G$. Then, $\langle S \rangle_t = [S]_t = [S]_{P_3}$ or $\langle S \rangle_t = V(G)$.
\end{lemma}

\begin{proof}
Let $S$ be a set of vertices of a prime graph $G$. If $S$ is not a clique, by Corollary~\ref{cor:hs-atom-tp}, $\langle S \rangle_t = V(G)$. Then, assume that $S$ is a clique. Therefore, using Theorem~\ref{thm:atom-convex}, we conclude that $[S]_t = [S]_{P_3} = S \cup \{u : u \not\in S$ and $|N(u) \cap S| \geq 2 \}$. Again, if $[S]_t$ is not a clique, then $\langle S \rangle_t = V(G)$. Hence, assume that $[S]_t$ is a clique. If $[S]_t$ is $t$-convex set, we are done. Otherwise, there is a vertex $u \not\in [S]_t$ having two neighbours in $[S]_t$. This means that $u  \in \langle S \rangle_t$ and $u$ is not adjacent to some vertex of $v \in S$. Since $\{u,v\}$ is a $t$-hull set of $G$, we have $\langle S \rangle_t = V(G)$.
\end{proof}

\begin{corollary} \label{cor:alg-prime}
Let $S$ be a set of vertices of a prime graph $G$ of order $n$. Then, one can test if $S$ is $t$-convex or compute its $t$-convex hull in $O(n^2)$.
\end{corollary}


Using these results, one can enumerate all $t$-convex sets of a prime graph efficiently, as shown in the following algorithm.

\begin{algorithm}[h]

\SetKwInOut{Input}{input}

\Input{A prime graph $G$}

${\cal C} \leftarrow \{\varnothing, V(G) \} \cup \{ \{ u \} : u \in V(G)\},$

$H \leftarrow E(G)$

\While{$H \neq \varnothing$}{

    $uv \leftarrow$ some element of $H$

    $S \leftarrow \{u,v\} \cup \{w : u,v \in N(w) \}$
    
    \If{$S$ is $t$-convex} {
    
        ${\cal C} \leftarrow {\cal C} \cup S$

    }
    remove from $H$ the edges of $G[S]$
}

\Return ${\cal C}$
    
\caption{Enumerating all $t$-convex sets of a prime graph.\label{alg:convexity-atom}}

\end{algorithm}

A direct analysis of Algorithm~\ref{alg:convexity-atom} runned on a prime graph of order $n$ and size $m$ leads to a time complexity equals $O(n^2m)$, because the loop of lines 3 to 8 is repeated $m$ times, line 5 has complexity $O(n)$, and line 6 can be done in $O(n^2)$ time (Corollary~\ref{cor:alg-prime}).

\section{Convexity number} \label{sec:convexity}

In this section we present an algorithm for finding a maximum proper $t$-convex set of any graph. This is an unexpected result, since this parameter is NP-complete for general graphs in the monophonic convexity~\cite{DPS-mhull}, for bipartite graphs in the geodetic convexity~\cite{DPRS-2012}, and for split graphs in the $P_3$ convexity~\cite{Centeno2010}.
We begin showing relating $t$-convex sets and $mp$-sugraphs.
Given a $t$-convex set $S$ and an $mp$-subgraph $F$ of a graph $G$, the set $S \cap V(F)$ can be or not a $t$-convex set of $G$. The following result shows that such set is a $t$-convex set of $F$.

\begin{lemma} \label{lem:SF}
Let $S$ be a $t$-convex set of a graph $G$ and $F$ an $mp$-subgraph of $G$. Then $V(F) \cap S$ is a $t$-convex set of $F$.
\end{lemma}

\begin{proof}
Let $S,F,$ and $G$ as in the statement of the corollary.
Denote $C = V(F) \cap S$. If $C = V(F)$ or $C = \varnothing$, we are done. Then, suppose for contradiction that $\varnothing \subset C \subset V(F)$ and that $C$ is not a $t$-convex set of $F$. This implies that there is a triangle path $P$ of $F$, not contained in $C$, joining two vertices of $C$.
It is clear that $P$ is also an induced path of $G$. This implies that $S$ is not a $t$-convex set of $G$, a contradiction.  
\end{proof}

For every $mp$-subgraph of a graph, we define a function on its $t$-convex sets and show how to combine them to obtain the $t$-convexity number of the graph. In the following definitions,
denote by ${\cal F}(G)$ the family of $mp$-subgraphs of $G$ and $F \in {\cal F}(G)$.

\begin{itemize}
\item For every $C \in {\cal C}_t(F)$, denote $F(C) = C$ union the vertices of the connected components of $G - C$ not containing $V(F) \setminus C$;

\item $c(G)= {\max} \{ \max \{\{ |F(C)| : C \in {\cal C}_t(F) \setminus \{V(F)\}\} : F \in {\cal F}(G) \}\}$.
\end{itemize}

\begin{theorem} \label{thm:convexity}
The $t$-convexity number of a connected graph $G$ is $c(G)$.
\end{theorem}

\begin{proof}
Let $F$ be an $mp$-subgraph of a graph $G$ and $C$ a $t$-convex set of $F$ such that $c(G) = |F(C)|$.
We have to show that $S = F(C)$ is a $t$-convex set of $G$.
Observe that, by the construction of $S$, the only vertices of $S$ that can have neighbours outside $S$ are the vertices in $C$. Since $C$ is a $t$-convex set of $F$, by Theorem~\ref{thm:atom-convex}, $C$ is a clique. Therefore, there are no two non-adjacent vertices of $S$ having neighbours in a same connected component of $G - S$.
Theorem~\ref{thm:atom-convex} also says that every vertex outside $C$ has at most one neighbour in $C$. Then, using Theorem~\ref{thm:tp-convexset}, we conclude that $S$ is a $t$-convex set of $G$.

Conversely, let $S$ be a maximum $t$-convex set of $G$. We will show that $c(G) \geq |S|$.
If $|S| = 1$, it is trivial. Then, assume $|S| \geq 2$.
Since $S \subset V(G)$ and $G$ is connected, there is an $mp$-subgraph $F$ of $G$ such that $|2| \leq |S \cap V(F)| < |V(F)|$.
Denote $C = S \cap V(F)$. By Lemma~\ref{lem:SF}, $C$ is a $t$-convex set of $F$.

We claim that there is no other $mp$-subgraph with this property. Then, suppose for contradiction that $F'$ is an $mp$-subgraph different of $F$ such that $|2| \leq |S \cap V(F')| < |V(F')|$.
By Theorem~\ref{thm:41}, 
we can consider $F = F_i$, $F' = F_j$, $i > j$. Then $F$ and $F'$ belong to distinct $B$-components of $G$, for $B = R_i$.
Furthermore, $B \cap S \neq \varnothing$.
If $B \setminus S = \varnothing$, the union of $S$ with the vertices of the $B$-component containing $F'$ that are not in $S$ would be a $t$-convex set of $G$ with more vertices than $S$.
Therefore $B \setminus S \neq \varnothing$. This implies that $|B \cap S| = 1$, because the fact that $S$ is $t$-convex implies that if $S$ contains two vertices of some clique, it contains all vertices of that clique. Write $\{w\} = B \cap S$.
By the choice of $B$, $w$ has neighbours outside $F$. Let $F''$ be the $mp$-subgraph containing $w$ that appears in the same $B$-component that $F'$.
Observe that $F''$ satisfies $|2| \leq |S \cap V(F'')| < |V(F'')|$ because, otherwise, the union of $S$ with the vertices of the $B$-component containing $F'$ that are not in $S$ would be a $t$-convex set of $G$ with more vertices than $S$.
This also implies that there is $u \in (B \setminus S) \cap V(F'')$.
Denote $C'' = S \cap V(F'')$.

Now, observe that, by the choice of $F$ and $F''$, there is a vertex $v \in C \setminus V(F'')$ and a vertex $v'' \in C'' \setminus V(F)$, both different of $w$.
Since $F$ and $F''$ are $mp$-subgraphs, there is a triangle path $P$ from $v$ to $u$ in $F - (C \setminus \{v\})$ and a triangle path $P''$ from $v''$ to $u$ in $F'' - (C \setminus \{v''\})$. It is clear that these two paths have only $u$ in common and that their concatenation form a triangle path joining two non-adjacent vertices of $S$ containing vertices outside $S$, which is a contradiction, then the claim holds.

To conclude the proof it suffices to observe that $S$ is a subset of $F(C)$.
\end{proof}

The algorithm for determining the $t$-convexity number of a general graph, Algorithm~\ref{alg:tconvexitynumber}, follows directly from Theorem~\ref{thm:convexity}. 

\begin{algorithm}[h]

\SetKwInOut{Input}{input}
\SetKwInOut{Output}{output}

\Input{A graph $G$}

${\cal F} \leftarrow$ find the family of $mp$-subgraphs of $G$

$c(G) \leftarrow 1$

\For{$F \in {\cal F}$}{

    ${\cal C}_t \leftarrow$ find the $t$-convex sets of $F$
    
    \For{$C \in {\cal C}_t$}{

         \If{$F(C) > c(G)$} {
    
              $c(G) \leftarrow F(C)$
         }
    }
}

\Return $c(G)$

\caption{$t$-Convexity number \label{alg:tconvexitynumber}}

\end{algorithm}

The computational complexity of Algorithm~\ref{alg:tconvexitynumber} is discussed in the sequel.
Line 1 can be done in $O(nm)$ time~\cite{Leimer1993} and the size of ${\cal F}(G)$ is smaller than $n$~\cite{Leimer1993,Tarjan1985}. Then, the loop of lines 3 to 7 is executed at most $n-1$ times. Line 4 has time complexity $O(n^2m)$ using Algorithm~\ref{alg:convexity-atom}.
Since the number of non-trivial $t$-convex sets of a prime graph is less than $m$, te loop of line 5 is executed $O(m)$ times.
Since line 6 can be done in time $O(m)$, the total complexity is $O(n^3m + nm^2) = O(nm^2)$.

\section{Hull number} \label{sec:hullnumber}

In this section we present a characterization of hull sets of a graph in the triangle path convexity and show how to use it for finding a minimum hull set of the graph. We begin considering the case where the graph is prime.

\begin{corollary}
The $t$-hull number of a non-trivial prime is two.
\end{corollary}

\begin{proof}
If $G$ is a complete graph, any two vertices form a $t$-hull set of the graph.
Otherwise, the result follows from Corollary~\ref{cor:hs-atom-tp}.
\end{proof}

Consider an $mp$-subgraph $F$ of a graph $G$ and a set $S$ of vertices of $G$. We say that $v \in V(F)$ is a {\em pivot of $F$} if there is an $mp$-subgraph $F'$ of $G$ such that $v \in V(F')$ and there is a vertex of $S \setminus V(F)$ in the $(V(F) \cap V(F'))$-component containing $F'$. 
We say that $S$ {\em satisfies $F$} if at least one of the following conditions holds:

\begin{enumerate}[{\em Condition} $1$:]
  \item There are two pivots in $F$ forming a $t$-hull set of $F$;

  \item There is a pivot $u \in V(F)$ contained in an $mp$-subgraph $F'$ and a vertex $v \in S \cap (V(F) \setminus V(F'))$ such that $\{u,v\}$ is a $t$-hull set of $F$;

  \item $S \cap V(F)$ is a $t$-hull set of $F$.
\end{enumerate} 

The next result characterizes the $t$-hull sets of a graph.

\begin{theorem} \label{thm:satisfies}
A set of vertices $S$ is a $t$-hull set of a reducible graph $G$ if and only if $|S| \geq 2$ and $S$ satisfies all $mp$-subgraphs of $G$.
\end{theorem}

\begin{proof}
Let $G$ be a reducible graph. First, consider a $t$-hull set $S$ of $G$. It is clear that $|S| \geq 2$. Suppose for contradiction that $F$ is an $mp$-subgraph of $G$
such that none of the three conditions are satisfied in $F$ by $S$. Denote $D = S \cap V(F)$.

If $D \neq \varnothing$, then since $S$ does not satisfy Condition 3 in $F$, $D$ is not a $t$-hull set of $F$ and, by Corollary~\ref{cor:hs-atom-tp}, $\langle D\rangle_t^F$ is a clique. But, since $S$ is a $t$-hull set of $G$, exists one vertex $w \in S \setminus V(F)$ such that there is a triangle path from $w$ to some vertex in $\langle D\rangle_t^G$ containing a vertex $u \in V(F) \setminus \langle D\rangle_t^G$.
It is clear that we choose $u$ in $V(F) \cap V(F')$ for some $mp$-subgraph $F'$, i.e., $u$ is a pivot of $F$.
If $u$ is not adjacent to some vertex $v$ of $D$, then $\{u,v\}$ would be a $t$-hull set of $F$. Which would mean that Condition 2 is satisfied in $F$.
Then, $u$ is adjacent to all vertices of $D$. This implies that $|D| = 1$, because otherwise $u$ would belong to $\langle D \rangle_t^F$.
Write $D = \{v\}$. We know that $\{u,v\}$ is not a $t$-hull set of $F$, because Condition 2 is not satisfied in $F$ by $S$. Then, $F$ contains at least two pivots. Further, we can say that $v$ is adjacent to all pivots of $F$. Denote by $D'$ the set formed by $v$ union all pivots of $F$. Since $D'$ is not a $t$-hull set of $F$, every two pivots of $F$ are adjacent. But, observe that this implies that $S$ is not a $t$-hull set of $G$. Then, $D'$ is a $t$-hull set, which would mean that $F$ contains two non-adjacent pivots and that $S$ satisfies Condition 1.
Then, $D = \varnothing$.

Let $F_1, \ldots, F_k$ be the $mp$-subgraphs of $G$ containing pivots of $F$. Denote $C_i = V(F) \cap V(F_i)$ for $i \in k$ and $C' = \underset{i \in [k]}{\bigcup} C_i$. Observe that, since $S$ is a $t$-hull set of $G$, we have $k \geq 2$. Further, the fact that $S$ does not satisfy Condition $1$ in $F$ implies that no $C_i$ is a $t$-hull set of $F$ for $i \in [k]$.
his also implies that $C'$ is not a $t$-hull set of $F$. Observe that this means that $S$ is not a $t$-hull set of $G$, a contradiction.

Conversely, consider a set $S \subseteq V(G)$ satisfying all $mp$-subgraphs of $G$ with $|S| \geq 2$.
We show that, for every $mp$-subgraph $F$ of $G$, it holds $V(F) \subseteq \langle S \rangle_t^G$.
If $S$ satisfies Condition~3 in $F$, it is clear that $V(F) \subseteq \hs_t^G$.

Now, suppose that $S$ satisfies Condition~2 in $F$. Let $u$ a pivot of $F$, $v \in S \cap V(F) \setminus V(F')$, where $F'$ is an $mp$-subgraph different of $F$ containing $u$, and $w \in S \setminus V(F')$ a vertex in the $(V(F) \setminus V(F'))$-component containing $F'$. Observe that there is a triangle path $P$ in $G$, joining $v$ to $w$, passing through $u$. Since $\{u,v\}$ is a $t$-hull set of $F$, we have that $V(F) \subseteq \hs_t^G$.

Finally, it remains the case in which $S$ satisfies only Condition 1 in $F$.
If $F$ contains two pivots $u_1$ and $u_2$ and there are two distinct $mp$-subgraphs $F_1$ and $F_2$ different of $F$ such that $u_1 \in V(F_1) \setminus(F_2)$ and $u_2 \in V(F_2) \setminus (F_1)$. Then, it is clear that $u,v \in [S]_t^G$, which implies that $V(F) \subseteq \hs_t^G$.
Then, there is an $mp$-subgraph $F_1$ different of $F$ containing all pivots of $F$.

This implies that all vertices of $S$ are in the $(V(F) \cap V(F_1))$-component containing $F_1$.
If $S$ satisfies Condition $j_1$ in $F_1$, for $j_1 \neq 1$, or there are at least two other $mp$-subgraphs containing the pivots of $F_1$, we are done. 
Therefore, suppose that $S$ satisfies only Condition 1 in $F_1$ and exists an $mp$-subgraph $F_2$ different of $F_1$ containing all pivots of $F$. Observe that the vertices of $V(F) \cap V(F_1)$ are not pivots of $F_1$.
Repeating this analysis in $F_2$ and so on, we will obtain a sequence
$F_0, F_1, \ldots, F_k$ of $mp$-subgraphs, for some $k \geq 2$, where $F_0 = F$ such that $S$ satisfies only Condition 1 in $F_j$, all pivots of $F_j$ are in $F_{j+1}$, for $0 \leq j < k$, and $V(F_k) \subset \langle S \rangle_t^G$.
Now, it is easy to see that $V(F) \subset \langle S \rangle_t^G$, concluding the proof.
\end{proof}



The algorithm for finding a minimum $t$-hull set of a general graph is given below.

\begin{algorithm}[h]

\SetKwInOut{Input}{input}
\SetKwInOut{Output}{output}

\Input{A connected reducible $G$}

$S \leftarrow \varnothing$

${\cal F} \leftarrow D$-ordered family of the $mp$-subgraphs of $G$

${\cal R} \leftarrow \{R_i = V(F_i) \cap ( V(F_1) \cup \ldots \cup V(F_{i-1}) ) : i \in [t], t = |{\cal F}|\}$

\For{$i = t$ \KwTo $2$} {

    $P \leftarrow$ pivots of $F_i$

    \If{$P \cup R_i$ is not a $t$-hull set of $F_i$} {
            let $v \in V(F_i) \setminus \langle P \cup R_i \rangle_t^{F_i}$ such that $\langle P \cup R_i \cup \{v\} \rangle_t^{F_i} = V(F_i)$
        
            $S \leftarrow S \cup \{v\}$
    }
}

\If{$V(F_1) \nsubseteq \langle S \rangle_t^G$} {
   
    \If{ exists $v \in V(F_1)$ such that $S \cup \{v\}$ is a $t$-hull set of $F_1$} {
        
            $S \leftarrow S \cup \{v\}$
    }
    \Else {
            let $u,v \in V(F_1)$ be a $t$-hull set of $F_1$
        
            $S \leftarrow S \cup \{u,v\}$
    }
}
\Return $S$

\caption{Hull number in the triangle path convexity \label{alg:tphn}}

\end{algorithm}

Now, we discuss the time complexity of Algorithm~\ref{alg:tphn}.
Let $G$ be a graph of order $n$ and size $m$.
Using the algorithm of~\cite{Leimer1993}, we can perform lines 2 and 3 in $O(nm)$ time. 
Since the number of $mp$-subgraphs is less than $n$~\cite{Leimer1993,Tarjan1985}, the number of iterations of the loop is $O(n)$.

Line 5 can be done in $O(m)$ time as follows. Define $G_i$ as the graph obtained by adding a vertex $v_i$ to $G$ adjacent to all vertices of $F_i$. Then, obtain a tree $T_i$ performing a breadth-first search rooted in $v_i$. Next, set $S_i = S$ and, for every vertex $u$ of $S_i$ not chosen yet, add to $S_i$ the neighbours of $u$ belonging to depth $d-1$ if $d$ is the depth of $u$ in $T$. It is easy to see that the set of pivots of $F_i$ is $S_i \cap R(G)$.

By Corollary~\ref{cor:alg-prime}, line 6 can be tested in $O(n^2)$ steps.
Line 7 can be done in constant $O(n)$. 
Since the remaining lines of the algorithm can clearly be performed in less time, the overall complexity for computing the $t$-hull number of a general graph is $O(n^3)$.


\newpage

\section*{Appendix}

Proof of Theorem~\ref{thm:tp-convexset}.

\begin{proof}
Let $S \subset V(G)$. If there is a vertex $v \not\in S$ such that $v$ has neighbours $u,w \in S$, then $uvw$ is a triangle path not contained in $S$. Then $S$ is not $t$-convex. Now, consider that there exist two non-adjacent vertices $u,v \in S$ such that $u' \in N(u)$, $v' \in N(v)$, and $u',v' \in V(C)$ for some connected component $C$ of $G - S$. If $u' = v'$, we come back to the first case. Then, we can assume that $u' \neq v'$. Consider a minimum path $P$ of $C$ joining $u'$ to $v'$. Since $C$ is a connected component of $G - C$, $P$ is an induced path of $G$. Now, note that $uPv$ is also an induced path, i.e., a triangle path of $G$. Then $S$ is not a $t$-convex set.

Assume now that $S$ is a $t$-convex set. Then, there exists a  
triangle path $u w_1 \ldots w_k v$ joining a pair of vertices $u,v \in S$ such that $k \geq 1$ and $w_i \notin S$ for every $i \in [k]$.
If $uv \in E(G)$, then $k = 1$, and $w_1$ is a vertex outside $S$ having two neihgbours in $S$.
Otherwise, since $P' = w_1 \ldots w_k$ is a path of $G - S$, all vertices of $P'$ belong to a same connected component of $G - S$. Then, $S$ has two vertices, namely, $u$ and $v$, having neighbours in a same connected component of $G - S$.
\end{proof}

Proof of Corollary~\ref{cor:atom}.

\begin{proof}
Let $G$ be a prime graph and suppose for contradiction that $S_1$ and $S_2$ are two $t$-convex sets of $G$ such that $\{u,v\} \subseteq S_1 \cap S_2$. Since $S_1 \neq S_2$, without loss of generality we can say that there is a vertex $w \in S_1 \setminus S_2$. Recall that, by Theorem~\ref{thm:atom-convex}, $S_1$ is a clique. Hence, $w$ is a vertex outside of a $t$-convex set, namely $S_2$, having two neighbours in $S_2$, $u$ and $v$, contradicting Theorem~\ref{thm:atom-convex}.
\end{proof}

Proof of Corollary~\ref{cor:numberprime}.

\begin{proof}
We use induction on $n$. If $n = 1$, it is trivial. Now, let $G$ be a prime graph of order $n$, for some $n \geq 2$, and suppose that every prime graph of order $n' < n$ has less than $n'$ non-trivial $t$-convex sets. If $G$ has no non-trivial convex sets with at least three vertices, we are done. Then, let $C$ be a non-trivial $t$-convex set of $G$ with at least three vertices and denote $G' = G - C$. Since $G$ is prime, $G'$ is also prime. Using the induction hypothesis and the fact that $|C| \geq 3$, we conclude that $G'$ has less than $n - 3$ non-trivial $t$-convex sets with at least three vertices.
Now, Theorem~\ref{thm:atom-convex} and Corollary~\ref{cor:atom} imply that every non-trivial $t$-convex set of $G$ with at least three vertices different of $C$ is also a non-trivial $t$-convex set of $G'$. Concluding the proof.
\end{proof}

Proof of Corollary~\ref{cor:alg-prime}.

\begin{proof}
Let $S$ be a set of vertices of a prime graph $G$ of order $n$. To test whether $S$ is a clique can be done in $O(n^2)$ and to test, for every vertex $u \not\in S$, whether $u$ has two neighbours in $S$ can be done in $O(n)$. Then, by Theorem~\ref{thm:atom-convex}, one can test if $S$ is $t$-convex in $O(n^2)$ time.

For computing the $t$-convex hull of $S$ we compute $S' = S \cup \{u : u \not\in S$ and $|N(u) \cap S| \geq 2 \}$. Next, as observed in the proof of Lemma~\ref{lem:car-prime}, if $S'$ is $t$-convex, then $\langle S \rangle_t = [S]_t$, otherwise $\langle S \rangle_t = V(G)$. It is clear that these tasks can be done in $O(n^2)$ steps.
\end{proof}

\begin{theorem}
Algorithm~{\em \ref{alg:convexity-atom}} is correct.
\end{theorem}

\begin{proof}
Let $G$ be a prime graph. We need to show that every $t$-convex set of $G$ will be added to ${\cal C}$ exactly once. All trivial $t$-convex sets of $G$ are identified at line 1. Since all non-trivial $t$-convex sets contain at least one edge, it is clear that if we consider every edge $uv \in E(G)$ once, to compute its $t$-convex hull, to check if it is different of $V(G)$, and add it to ${\cal C}$, we will generate all proper $t$-convex sets of $G$. However, it is possible that some of them have been generated more than once.

We claim that, for achieving the unicity, instead of computing the $t$-convex hull of $\{u,v\}$, we compute 
$S \leftarrow \{u,v\} \cup \{w : u,v \in N(w) \}$ (line 5) and after testing whether it is a $t$-convex set we remove all edges of $G[S]$ (line 8).

This clearly avoids that a proper $t$-convex set be generated twice.
It remains to show that it does not yield that some of them be not generated at least once.
Let $xy$ be an edge removed at line 8 in some iteration $i$ of the loop and denote by $S_i$ the set constructed at line 5 of this iteration.
On one hand, $S_i$ is a $t$-convex set of $G$. By Lemma~\ref{lem:car-prime}, there is no proper $t$-convex of $G$ different of $S_i$ containing both $x$ and $y$. Then, $xy$ can be removed from $H$ at this moment. On the other hand, $S_i$ is not a $t$-convex set of $G$. 
By Lemma~\ref{lem:car-prime} again, the only $t$-convex set containing both $x$ and $y$ is $V(G)$. Then, $xy$ can also be removed from $H$ at this moment.
\end{proof}

Next, we present the correctness of the algorithm for finding a minimum $t$-hull set of a general graph.

\begin{theorem}
Algorithm~$\ref{alg:tphn}$ is correct.
\end{theorem}

\begin{proof}
Let $G$ be a graph with $t \geq 2$ $mp$-subgraphs and $S$ the set obtained by the Algorithm~\ref{alg:tphn} runned over $G$.
Note that every $mp$-subgraph of $G$ is considered exactly once by the algorithm, $t-1$ of them in the loop of the algorithm and the remaining one in the Lines 9 to 14. Denote by $S_i, P_i$, and $v_i$ the corresponding instances of $S, P$, and $v$ at the end of iteration $i$, for $2 \leq i \leq t$.

First, we show that $S$ is a $t$-hull set of $G$. For this, we prove that $S$ satisfies every $mp$-subgraph of $G$.
Consider some iteration $i$, $2 \leq i \leq t$, of the loop.
Since the vertex chosen at line 7 is a vertex of $V(F_j) \setminus R_j$, for $2 \leq j \leq t$, and these vertices will no more be considered at line 7 in the future, we conclude that at moment of the execution of line 6 it holds $S_{j-1} \cap V(F_j) = \varnothing$.

Hence, if the algorithm does not run line 7, it means that $S_{i-1}$ satisfies Condition 1 in $F_i$ or
$|\langle S_{i-1} \rangle_t^G \cap R_i| \leq 1 < |R_i|$.
In the latter case, we claim that exists in $S \setminus S_i$ some vertex of some $R_i$-component of $G$ not containing $F_i$ and this will imply that $S$ satisfies Condition 1 in $F_i$.
Consider the case where $|\langle S_2 \rangle_t^G \cap R_i| \leq 1 < |R_i|$ and the algorithm reachs line 9.
Then, some vertex will be added to $S_2$ at line 11 or 14, because $F_1$ is contained in some of the $R_i$-components of $G$ not containing $F_i$ and no vertex of $S_i$ is in any of these $R_i$-components, because otherwise $\langle S_{i-1} \rangle_t^G \cap R_i = R_i$.
Then, let $u \in S \setminus S_i$. Observe that there is triangle path in $G$ from $u$ to some vertex of $S_i$ not passing trough $\langle S_{i-1} \rangle_t^G \cap R_i$. Then, $V(F_i) \subset \hs_t^G$ and $S$ satisfies Condition 1 in $F_i$.

Now, if the algorithm executes line 7, then either $S_i$ satisfies Condition 2 in $F_i$ or, for the same reason of the previous paragraph, there is in $S \setminus S_i$ some vertex of some $R_i$-component of $G$ not containing $F_i$. Consequently, $S$ satisfies Condition 2 in $F_i$.

If line 10 is not executed, then $S_2$ satisfies Condition 1 in $F_1$.
If line 11 is executed, then $S$ satisfies Condition 2 in $F_1$.
If line 14 is executed, then $S$ satisfies Condition 3 in $F_1$.
Then, all $mp$-subgraphs $G$ are satisfied by $S$ and $|S| \geq 2$, then by Theorem~\ref{thm:satisfies}, $S$ is a $t$-hull set of $G$.

It remains to prove that the $t$-hull number of $G$ is at least $|S|$.
For every iteration $i$ of the loop of the algorithm, define $C_i = V(F_i) \setminus \langle P_{i-1} \cup R_i \rangle_t^{F_i}$ if line 7 is executed, otherwise define $C_i = \varnothing$.
Observe that if $C_i \neq \varnothing$, then $C_i$ is a $t$-concave set of $G$. It is clear that every $t$-hull set must contain at least one vertex of each $t$-concave set. Now, observe that if $C_i$ and $C_j$ are both non-empty and $i \neq j$, then $C_i \cap C_j = \varnothing$. Which implies that the $t$-hull number of $G$ is at least $|S_2|$.

Finally, consider the execution of lines 9 to 14.
We observe that line 14 is executed only if $S_2 = \varnothing$.
If line 11 is not executed or if $|S_2| = 1$, we are done.
Then, consider $|S_2| \geq 2$ and let two vertices $v_i$ and $v_j$ chosen in iterations $i$ and $j$, respectively, for $j < i$. Recall that our numbering of iterations is decreasing. Observe that $v_i \in V(F_i) \setminus V(F_j)$ and $v_j \in V(F_j) \setminus V(F_i)$.
By Theorem~\ref{thm:41}, $R_i \subset V(F_i)$ of $G$ such that 
$F_i$ and $F_j$ are contained in different $R_i$-components of $G$.
It is clear that $P_i \cup R_i \subset \langle S_i \cup \{v_j\} \rangle_t^G$.
Then, $V(F_i) \subset \langle S_2 \rangle_t^G$.
This implies that all $t$-concave sets $C_k$ defined in the previous paragraph are contained in $\langle S_2 \rangle_t^G$. Since the vertex chosen in line 11 form a $t$-hull set of $G$, it is chosen outside $\langle S_2 \rangle_t^G$, i.e., of a $t$-concave set of $G$ different of the ones already considered. Which means that the $t$-hull number of $G$ is at least $|S|$.
\end{proof}

\end{document}